\documentclass[12pt]{article}
\usepackage[top=0.75in,bottom=0.75in,left=1in,right=1in]{geometry}

\usepackage{graphicx,psfrag,epsf,color}
\usepackage{epsfig, graphicx}
\usepackage{hyperref}
\usepackage{multirow}
\usepackage{booktabs}
\usepackage{amsthm,amsmath,amssymb}
\usepackage{natbib}
\usepackage{mathtools}
\usepackage[linesnumbered, ruled]{algorithm2e}

\newcommand{\yifan}[1]{\textcolor{black}{#1}}

\newcommand{\yf}[1]{\textcolor{black}{#1}}

\newtheorem{remark}{Remark}

\newtheorem{definition}{Definition}[section]

\newtheorem{theorem}{Theorem}[section]
\newtheorem{lemma}{Lemma}[section] 
\newtheorem{proposition}{Proposition}[section]

\newtheorem{corollary}{Corollary}[section]

\newcommand{\bdy}{{y}}
\newcommand{\bdu}{{u}}
\newcommand{\bdU}{{U}}

\newcommand{\bdY}{{Y}}
\newcommand{\bdG}{{G}}
\newcommand{\bdA}{{A}}
\newcommand{\bdtheta}{{{\theta}}}

\title{\yifan{Semiparametric fiducial inference for Cox models}}
\author{
Yifan Cui\thanks{Center for Data Science, Zhejiang University, China. Correspondence to \href{mailto:cuiyf@zju.edu.cn}{cuiyf@zju.edu.cn}}~
Jan Hannig\thanks{Department of Statistics and Operations Research, UNC-Chapel Hill, U.S.}~
Paul Edlefsen \thanks{Vaccine and Infectious Disease Division, Fred Hutch, U.S.}}
\date{}

\begin{document}

\maketitle

\abstract{R. A. Fisher introduced the fiducial distribution as a potential replacement for the Bayesian posterior distribution in the 1930s. During the past century, fiducial approaches have been explored in various parametric and nonparametric settings.
However, to the best of our knowledge, no fiducial inference has been developed in the realm of semiparametric statistics. 
In this paper, we propose a novel fiducial approach for semiparametric models.
\yf{In memory of Sir David Cox who passed away in 2022,} we use the Cox proportional hazards model, which is the most popular model for the analysis of survival data, as a running example. 
Other models and extensions are also discussed. 
In our experiments, we find that our method performs particularly well in situations where the maximum likelihood estimator fails.\\
}

\noindent {{\bf keywords:}
Bernstein–von Mises theorem, Conic optimization, Cox model, Fiducial inference, Gibbs sampler, Semiparametric model}

\section{Introduction}
\yifan{Fiducial inference was introduced by R. A. Fisher \citep{Fisher1930,Fisher1933} as a novel mode of statistical reasoning, distinct from both Bayesian and frequentist frameworks.} The fiducial distribution can be viewed as a potential replacement for the Bayesian posterior distribution
 in a data-driven sense which does not rely on a subjective prior selection. 

\yf{While in the second half of the 20th century fiducial distribution has been studied only sparingly, e.g., \cite{Wilkinson1977, dawid1982functional, Zabell1992, Barnard1993},}
there has been a fast-growing literature on fiducial methods and related approaches in the past two decades.  For example, \cite{WangYH2000, TaraldsenLindquist2013} showed how fiducial distributions naturally arise within a decision theory framework;
\cite{HannigLee2009} proposed fiducial solutions to wavelet regression and \cite{WandlerHannig2012b} addressed extreme value estimation from a fiducial perspective; \cite{hannig2009,hannig2016generalized} formalized the mathematical definition of generalized fiducial distribution. Their argument is based on inverting a data generating algorithm (DGA) that associates data to the parameters and a random component with a known distribution. The generalized fiducial distribution is then obtained by inverting the DGA for the parameter.
This formal definition facilitated the application of fiducial inference to a variety of important statistical problems such as linear mixed model \citep{CisewskiHannig2012}, ultrahigh-dimensional regression \citep{lai2015}, 
censored data and survival analysis \citep{chen2016generalized,cuihannig2019,cui2024unified}, 
model selection \citep{williams2019}, 
empirical Bayes estimation and $g$-modeling \citep{cui2024fiducial}, 
 vector autoregressive graph selection \citep{williams2023eas}, etc. 
as well as to other fields including psychology  
\citep{liu2016generalized,liu2017generalized,Liu2019,Neupert2019} and forensic science \citep{HANNIG2019572}.
We refer to \cite{murph2023introduction,dawid2024fiducial} for recent reviews of  fiducial inference.

Other related approaches include confidence distributions \citep{SinghXieStrawderman2005, XieSingh2013, claggett2014meta, schweder2016confidence, HjortSchweder2018}, Dempster-Shafer theory \citep{Dempster:1968vd,shafer1976mathematical,EdlefsenLiuDempster2009}, inferential models \citep{MartinLiu2013a,martin2015inferential,martin2015marginal,liu2020inferential}, objective Bayesian inference \citep{BergerBernardoSun2009, BergerBernardoSun2012}, repro methods \citep{xie2022repro,wang2022finite}, and structural inference \citep{DawidStoneZidek1973,Fraser1966}.
We refer to \cite{cui2025demystifying} for the connections between inferential models, confidence curves, and fiducial inference, and  \cite{cui2023confidence} for the connection between confidence distribution and fiducial inference.

While tremendous progress has been made in the area of the foundation of statistics in the past decades, there are few papers on semiparametric models \yifan{within the context of fiducial inference}. A semiparametric model is a statistical model that has parametric and nonparametric components \citep{bickel1993efficient,tsiatis2006semiparametric,kosorok2008introduction}.
In this paper, we apply fiducial inference to such models and propose novel inferential tools for the finite-dimensional parameter of interest. 
To our knowledge, this is the first time fiducial inference has been systematically applied to semiparametric models. Specifically, we consider the celebrated Cox proportional hazards model \citep{cox1972regression} as a running example. 
We propose a novel Gibbs sampler to sample from the derived fiducial distribution using conic optimization. Upon obtaining fiducial samples, we use the samples to construct statistical inference. For example, the median of the samples is used as a point estimator, and appropriate quantiles are used to construct confidence intervals.

We establish an asymptotic theory that verifies the frequentist validity of the proposed
fiducial approach. First, we prove the consistency of the proposed fiducial point estimator.
Next, we establish a Bernstein–von Mises theorem for the fiducial distribution. As a consequence of the Bernstein–von Mises theorem, the proposed confidence intervals provide asymptotically correct coverage, and the
proposed fiducial estimator is first-order asymptotically equivalent to the maximum partial likelihood estimator.
It is noteworthy that the proposed point estimator works well in scenarios when the classic maximum partial likelihood estimator fails.

The remainder of the article is organized as follows. In Section~\ref{sec:method}, we take a new look at the Cox model from a data generating perspective and derive our generalized fiducial distribution for our parameter of interest. \yf{We use the Cox proportional hazards model both because of its foundational role in survival analysis and in memory of Sir David Cox.} We then propose a novel conic optimization-based Gibbs sampler to sample from the fiducial distribution. In Section~\ref{sec:theory}, we develop consistency and asymptotic normality for the proposed fiducial estimator. In Section~\ref{sec:simulation}, we demonstrate the superiority of our estimator compared to the maximum likelihood estimator through simulation studies. Section~\ref{sec:extension} provides several extensions of the proposed method to other semiparametric models. 
Section~\ref{sec:real} describes a real data application on modern HIV trials. The article concludes with a discussion of future work in Section~\ref{sec:discussion}. Additional results and proofs are provided in the Appendix and Supplementary Material.

\section{Methodology}\label{sec:method}

\subsection{The Cox proportional hazards model revisited}
We consider the Cox proportional hazards model which is the
most popular model \citep{cox1972regression} for the analysis of survival data. 
Suppose we observe right-censored survival data $(X_i,Y_i,\Delta_i)$, $i=1,\ldots,n$, where $Y_i=\min\left \{T_i, C_i\right \}$, censoring indicator $\Delta_i=I\left\{T_i\leqslant C_i\right\}$, and $T_i$, $C_i$, $X_i$ represent the failure time, the censoring time, and explanatory covariates for the $i$-th subject, respectively. We assume a noninformative censoring mechanism here, i.e., $T_i$ and $C_i$ are independent given $X_i$. \yf{We direct the reader to Chapter 11 of \cite{james2013introduction} for an extensive overview of concepts in survival analysis.}

\yf{Let $S_i(t)= P(T_i > t)$ and $\lambda_i(t)=\lim_{\Delta t \to 0} {P(t \leq T_i < t+\Delta t \mid T_i \geq t)}/{\Delta t}$} denote the survival function and hazard function of the $i$-th subject, respectively, and $\mathcal R_i$ denote the at-risk index set at the time $Y_i$, i.e., $\mathcal R_i=\{j: Y_j\geq Y_i\}$. Note that the full likelihood function is
\begin{align*}
  \prod_{i=1}^n \left[ \frac{\lambda_i(Y_i)}{\sum_{j\in \mathcal R_i} \lambda_j(Y_i)} \right]^{\Delta_i}  \left[\sum_{j\in \mathcal R_i} \lambda_j(Y_i)\right]^{\Delta_i} S_i(Y_i).
\end{align*}

The Cox proportional hazards model posits the following form of hazard function:
\begin{align*}
\lambda_i(t)= \lambda_0(t) g(\beta^\top X_i),
\end{align*}
where $\lambda_0(t)$ is a baseline hazard function, and $g(\beta^\top X)$ is a link function. We first consider the most prevalent case $g(\beta^\top X)=\exp(\beta^\top X)$. This leads to the following likelihood function
\begin{equation}\label{eq:CoxCompleteL}
  L(\beta,\lambda_0)= \prod_{i=1}^n \left[ \frac{\exp(\beta^\top X_i)}{\sum_{j\in \mathcal R_i}\exp(\beta^\top X_j)} \right]^{\Delta_i}   
\left[\lambda_0(Y_i){\sum_{j\in \mathcal R_i}\exp(\beta^\top X_j)}\right]^{\Delta_i} S_i(Y_i),
\end{equation}
where $S_i(t)=\exp(-\Lambda_0(t)\exp(\beta^\top X_i))$ and $\Lambda_0(t)=\int_0^t \lambda_0(s)ds$.

We are interested in statistical inference on the log hazard ratio $\beta$. Under the Cox proportional hazards model,
the well-known Cox's partial likelihood \citep{cox1975partial} is given below
\begin{equation}\label{eq:CoxPartialL}
 \yifan{L_n(\beta)} = \prod_{i=1}^n \left[ \frac{\exp(\beta^\top X_i)}{\sum_{j\in \mathcal R_i}\exp(\beta^\top X_j)} \right]^{\Delta_i}.
\end{equation}
\yifan{In the Cox proportional hazards model, there is no parametric assumption on the baseline hazard $\lambda_0(t)$ and the inference about $\beta$ relies only on the relative hazard $\exp(\beta^\top X)$. It was the brilliant insight of \cite{cox1975partial} that one can use only the partial likelihood $L_n(\beta)$ ignoring portions of the full likelihood without much loss of efficiency in inference on $\beta$, as long as there are no \yf{smoothness} assumptions on $\lambda_0(t)$.} Maximizing the log partial likelihood can be achieved via the Newton-Raphson algorithm, and the inverse of the Hessian matrix evaluated at the maximum likelihood estimator $\tilde \beta$ can be used for constructing confidence intervals. The consistency and asymptotic normality of $\tilde \beta$ can be established \citep{andersen1982cox,lin1989robust,fleming2013counting}. However, despite the popularity of the Cox model and partial likelihood approach, \yf{the maximum likelihood estimator sometimes does not converge} when the sample size is relatively small as seen in our simulation section. In comparison, we will see that the proposed fiducial-based estimation works well.

\subsection{Generalized fiducial inference revisited}\label{s:FiducialDef}

Generalized fiducial inference begins with a DGA that relates the data $\bdY$ to the parameter of interest $\bdtheta$:
\begin{equation} \label{data-generating}
  \bdY=\bdG_{\bdtheta}(\bdU),
\end{equation}
where $\bdG_{\bdtheta}$ is a deterministic function, and $\bdU$ is an auxiliary random variable whose distribution is completely known and does not depend on the unknown parameter $\bdtheta$. As in the definition of the likelihood function, generalized fiducial inference proceeds by switching the roles of the data and the parameter: once the data $\bdy$ are observed, \eqref{data-generating} is viewed as an equation to be inverted for $\bdtheta$. Accordingly, we define
\begin{equation}\label{data-generating-inv}
Q_{\bdy}(\bdu)
=
\arg\min_{\bdtheta\in\Theta}
\rho\!\left(\bdG_{\bdtheta}(\bdu),\bdy\right),
\end{equation}
where $\rho$ is a smooth semi-metric on the sample space. Here, $Q_{\bdy}(\bdu)$ should be understood as a selected minimizer when the argmin set contains more than one element.

Ideally, the minimum in \eqref{data-generating-inv} is equal to zero, in which case $Q_{\bdy}(\bdu)$ provides an exact inverse of the DGA. In general, two difficulties may arise. First, for a given pair $(\bdy,\bdu)$, the minimizer in \eqref{data-generating-inv} may not be unique. This can be handled either through imprecise-probability methods \citep{Dempster2008} or, in standard applications, by selecting one solution according to a fixed or random rule; see \cite{Hannig2013}. Second, there may be no $\bdtheta$ such that $\bdy=\bdG_{\bdtheta}(\bdu)$. To describe the auxiliary values for which an exact solution exists, define
\begin{equation}\label{U0}
\mathcal U_{\bdy,0}
=
\left\{
\bdu:
\bdy=\bdG_{\bdtheta}(\bdu)
\text{ for some }\bdtheta\in\Theta
\right\}.
\end{equation}
If $P(\bdU\in\mathcal U_{\bdy,0})>0$, then one may sample $\bdU$ from its distribution truncated to $\mathcal U_{\bdy,0}$ and map the resulting draws through $Q_{\bdy}$. For continuous models, however, $P(\bdU\in\mathcal U_{\bdy,0})=0$ typically, so conditioning directly on $\mathcal U_{\bdy,0}$ is ill-defined because of the Borel paradox \citep{CasellaBerger2002}.

To address this issue, we enlarge the conditioning event by introducing a small tolerance $\varepsilon$:
\begin{equation}\label{Ueps}
\mathcal U_{\bdy,\varepsilon}
=
\left\{
\bdu:
\min_{\bdtheta\in\Theta}
\rho\!\left(\bdG_{\bdtheta}(\bdu),\bdy\right)
\le \varepsilon
\right\}.
\end{equation}
Let $\bdU_\varepsilon^*$ denote a random variable distributed as $\bdU$ truncated to $\mathcal U_{\bdy,\varepsilon}$, and let $P^*_\varepsilon$ be the distribution of $Q_{\bdy}(\bdU_\varepsilon^*)$.

\begin{definition}\label{def:limitGFI}
A probability measure $P^*$ on the parameter space $\Theta$ is called a generalized fiducial distribution (GFD) if
$
P^*_\varepsilon
$ converges weakly to $P^*$
as $\varepsilon\to0$. 
\end{definition}
For discrete models, $P(\bdU\in\mathcal U_{\bdy,0})>0$, so $P^*$ is a distribution of any element selected from the random set $Q_{\bdy}(\bdU_0^*)$. For continuous models satisfying suitable regularity conditions, \cite{hannig2016generalized} showed that the limiting GFD in Definition~\ref{def:limitGFI} has density
\begin{equation}\label{eqn:GFD}
r(\bdtheta)
=
\frac{
f(\bdy;\bdtheta)\,J(\bdy,\bdtheta)
}{
\int_{\Theta} f(\bdy;\bdtheta')\,J(\bdy,\bdtheta')\,d\bdtheta'
},
\end{equation}
where $f(\bdy;\bdtheta)$ is the likelihood function and
\begin{equation}\label{eq:ContJacobian}
J(\bdy,\bdtheta)
=
D\!\left(
\nabla_{\bdtheta}\bdG_{\bdtheta}(\bdu)
\bigg|_{\bdu=\bdG_{\bdtheta}^{-1}(\bdy)}
\right).
\end{equation}
Here $\bdG_{\bdtheta}^{-1}(\bdy)$ denotes the value of $\bdu$ satisfying $\bdy=\bdG_{\bdtheta}(\bdu)$, and $D(\cdot)$ depends on the choice of $\rho$. When $\rho(\bdy,\bdy^*)=\|\bdy-\bdy^*\|_2^2$, then
$
D(\bdA)=\left(\det(\bdA^\top\bdA)\right)^{1/2},
$
which is the product of the singular values of $\bdA$.

{In what follows, we will be using stars, e.g., $\beta^*$, to denote random variables with respect to the GFD.}

\subsection{A data generating perspective for the Cox model}
\label{s:basicDGA}

We now look at the Cox model from the following novel data generating perspective. \yf{We describe a particular DGA \eqref{data-generating} that can be used to generate realizations of $(Y_i,\Delta_i)$ following the full likelihood \eqref{eq:CoxCompleteL}. In the next section, we derive the corresponding GFD of $(\beta,\lambda_0)$, with particular emphasis on the marginal distribution of $\beta$, and observe that this marginal fiducial distribution is related to Cox's partial likelihood approach.}

\yf{The input of our DGA is the potential censoring times, $c_1,\ldots,c_n$, that are generated from a separate censoring mechanism. The algorithm generates the failure times $t_1,\ldots,t_m$ and indicators of which subjects failed at which failure times $i_1,\ldots,i_m$; here and in what follows $m$ denotes the number of failures. These then are relabeled and combined into the data $(X_i,Y_i,\Delta_i)$, $i=1,\ldots,n$ actually observed by the investigator.} 

We first model the first failure. Set
$\mathcal R_{1}(t) = \{1,\ldots,n\}\setminus\{i : c_i<t\}$. The hazard that any subject fails is the sum of hazards of all at-risk subjects $\bar\lambda
_1(t) = \sum_{j\in \mathcal R_1(t)}\lambda_j(t) =\lambda_0(t) \sum_{j\in \mathcal R_1(t)} \exp(\beta^\top X_j)$. The corresponding $\bar\Lambda_1(t)=\int_{0}^{t} \bar\lambda
_1(s)\,ds$ and $\bar S_1(t)=\exp(-\bar\Lambda_1(t))$. 
We generate the time of the first failure by 
$t_1=\bar S_1^{-1}(W_1)$, where $W_1$ is generated from a uniform distribution on $(0,1)$. 
The subject that failed first is generated as $i_1\sim \mbox{Multinomial}(1,\vec{q}_1)$, where
\begin{equation*}
    \vec{q}_1 = [q_{1,1}(\beta),\ldots,q_{1,n}(\beta)]^\top,
    \mbox{ and }
    q_{1,i}(\beta)=\begin{cases}
       \frac{\exp(\beta^\top X_i)}{\sum_{j\in \mathcal R_1(t_1)} \exp(\beta^\top X_j)}
       & \mbox{ if $i\in \mathcal R_{i_1}= \mathcal R_1(t_1)$};\\
       0 & \mbox{ otherwise}.
    \end{cases}
\end{equation*}
Because $t_1$ is the smallest failure time, any subject $i$ with censoring time $c_i<t_1$ is a censored observation with censoring time $c_i$. 

We continue generating subsequent failures by repeating the above step.
At the time of the $(k-1)$-th failure, we have $t_{k-1}$ and $i_1,\ldots, i_{k-1}$. The at-risk set after the $k-1$-th failure is
\begin{equation}\label{eq:atRiskSet}
\mathcal R_{k}(t) = \{1,\ldots,n\}\setminus\left(\{i: c_i<t\}\cup \{i_1,\ldots, i_{k-1}\}\right).
\end{equation}
Define $\bar\lambda
_k(t) =\sum_{j\in \mathcal R_k(t)} \lambda_j(t)= \lambda_0(t) \sum_{j\in \mathcal R_k(t)} \exp(\beta^\top X_j)$,  $\bar\Lambda_k(t)=\int_{t_{k-1}}^{t\vee t_{k-1}} \bar \lambda
_k(s)\,ds$, and $\bar S_k(t)=\exp(-\bar \Lambda_k(t))$. 
We generate the time of the $k$-th failure by 
$t_k=\bar S_k^{-1}(W_k)$, where $W_k$ is generated from the $\text{Uniform}(0,1)$ distribution. 
The $k$-th subject that failed is generated as $i_k\sim \mbox{Multinomial}(1,\vec{q}_k)$, where
\begin{equation}\label{eq:Qkl}
    \vec{q}_k = [q_{k,1}(\beta),\ldots,q_{k,n}(\beta)]^\top,
    \mbox{ and }
    q_{k,i}(\beta)=\begin{cases}
       \frac{\exp(\beta^\top X_i)}{\sum_{j\in \mathcal R_k(t_k)} \exp(\beta^\top X_j)}
       & \mbox{ if $i\in \mathcal R_{i_k}= \mathcal R_k(t_k)$};\\
       0 & \mbox{ otherwise}.
    \end{cases}
\end{equation}
Again subjects $i\in\mathcal R_{k}(t_{k-1})=\mathcal R_{k-1}\setminus \{i_{k-1}\}$ with censoring time $c_i<t_k$ are censored observations with censoring time $c_i$.
This is repeated until either  the generated failure time $t_k=\infty$ 
(the last observation is censored) or 
${\mathcal R}_{k+1}(t_{k})=\emptyset$ (the last observation is a failure). 

The above procedure completes the data generation given the censoring times.
As shown in the following proposition, the above data generating mechanism produces data from the Cox proportional hazards model.
\begin{proposition}\label{prop:compatible}
\yf{Data  $(X_i,Y_i,\Delta_i)$, $i=1,\ldots,n$ generated by the above DGA have the joint distribution corresponding to likelihood \eqref{eq:CoxCompleteL}.}
\end{proposition}

\begin{remark}
Technically, the likelihood \eqref{eq:CoxCompleteL} does not allow for ties. However, ties can be introduced by allowing the cumulative hazard function $\Lambda_0(t)$ to have jumps and changing the likelihood accordingly. The details are shown in Appendix~\ref{s:generalDGA}, where we also propose an approximation to the DGA that leads to fiducial distribution which is analogous to Breslow partial likelihood \citep{peto1972contribution,breslow1974covariance}.

\end{remark}



\subsection{Fiducial inversion for the Cox model}
\label{s:DGAinversion}
In this subsection, we derive the GFD of the Cox proportional hazards model by inverting the DGA from the previous section. Recall that we denote the number of failures by $m$ and assume $m>0$. Note that the multinomial variables $i_k, k=1,\ldots, m,$ depend only on the parameter $\beta$, while the failure times $t_k$ depend on both $\beta$ and $\lambda_0(t)$. \yf{Consequently, we can find the inverse \eqref{data-generating-inv} by first inverting the multinomial part of the DGA into $\beta$ and then plugging the solution into the time part of the DGA and solving for $\lambda_0(t)$. The second part of this inversion will not modify the fiducial distribution of $\beta$ derived from the multinomial distribution, because there are no smoothness assumptions on $\lambda_0(t)$ and the solution of this second inversion always exists. Thus,} here we describe the former inversion while the latter is in Appendix~\ref{s:BaselineHazard}.

The GFD for $\beta$ is obtained by inversion of the multinomial distributions \eqref{eq:Qkl}. To this end, we follow \cite{lawrence2024new} as also explained in \cite{hannig2016generalized}. 
Since all $m_k=1$, i.e., all event times are distinct, the fiducial distribution 
for $\beta$ is based on solving the inequalities, 
\begin{equation}\label{eq:LawrenceInv2}
    U_k^*\leq \frac{\exp(\beta^\top X_{i_k})}{\sum_{j\in \mathcal R_{i_k}} \exp(\beta^\top X_j)}=q_{k,i_k}(\beta):=q_k(\beta), \quad k=1,\ldots,m,
\end{equation}
where $U_k^*$ are jointly uniform distribution on the set on which the solution to \eqref{eq:LawrenceInv2} exists, and
$i_k$ is the index of the subject who failed at the time of the $k$-th failure.

\yf{There is no closed form solution for the distribution of \eqref{eq:LawrenceInv2} and sampling $(U_1^*,\ldots,U_m^*)$ from the uniform distribution on the set on which solution to \eqref{eq:LawrenceInv2} exists is computationally challenging.} 
Therefore, we propose a Gibbs sampler for generating $U_k^*$ and $\beta_j^*$. Notice that for any $k=1,\ldots,m$, the distribution of $U_k^*$ given the rest is $\text{Uniform}(0,q_k^*)$, where the upper bound $q_k^*$ is found by the following constraint optimization problem:
\begin{align}\label{eq:qk*}
    q_k^*=\max_{\beta} \frac{\exp(\beta^\top X_{i_k})}{\sum_{j\in \mathcal R_{i_k}} \exp(\beta^\top X_j) }
\end{align}
subject to the following $(m-1)$ constraints, 
\begin{align*}
   U_h^*\leq q_h(\beta),\quad h\in\{1,\ldots,m\}\setminus\{k\}.
\end{align*}
Upon obtaining $q_k^*$, we then update $U_{k}^* \sim \text{Uniform}(0,q_k^*)$. 

After the $j$-th cycle of the Gibbs sampler is finished, we want to generate a representative of the set $Q(U^*)=\{\beta: \beta \mbox{ satisfying \eqref{eq:LawrenceInv2}}\}$.
To this end, generate $w$ from a standard normal distribution and then calculate $\beta_j^*$ by solving 
\begin{align}\label{eq:qk**}
\max_{\beta} \beta^\top w
\end{align}
subject to the following $m$ constraints,
\begin{equation}\label{eq:Qfeasable}
   U_h^*\leq q_h(\beta),\quad h\in\{1,\ldots,m\}.
\end{equation}
Note that if \eqref{eq:qk**} fails to converge due to unbounded constraints, one can use $-w$ instead of $w$. 
 After the $\beta_j^*$ have been generated, one can also generate a sample from the fiducial distribution of the baseline hazard. Details are given in Appendix~\ref{s:BaselineHazard}.

The proposed Gibbs sampler for the log hazard ratio $\beta$ is summarized in  Algorithm~\ref{alg:gibbs}. 
Once the fiducial sample is generated, 
we propose to use the median of $\beta_j^*$, $j=n_{\text{burn}}+1,\ldots,n_{\text{burn}}+n_{\text{mcmc}}$ as a point estimator, and the empirical 0.025 quantile as a lower limit and the empirical 0.975 quantile as an upper limit.

\begin{center} 
\begin{minipage}{\linewidth}
\begin{algorithm}[H] 
\SetAlgoLined
\caption{A fiducial Gibbs sampler \label{alg:gibbs}}
\KwIn{Dataset  $(X_i, Y_i, \Delta_i)$,
$n_{\text{mcmc}}$, $n_{\text{burn}}$.}
\ShowLn Use $q_1(\tilde \beta),\ldots,q_m(\tilde \beta)$ as initial values and generate $U_{k}^*$ by $\text{Uniform}(0,q_k(\tilde \beta))$, where $\tilde\beta$ is the MLE. If the MLE did not converge, set $\tilde\beta=0$.\\
\For{$j = 1$ \textbf{to} $n_{\text{burn}}+n_{\text{mcmc}}$}
{ 
\For{$k = 1$ \textbf{to} $m$}
{
\ShowLn Solve $q^*_k$ by~\eqref{eq:qk*} based on the current $U_{h}^* $, $h=1,\ldots, k-1, k+1 \ldots,m$\;
\ShowLn Update $U^*_k$ by a random sample $U_{k}^* \sim \text{Uniform}(0,q_k^*)$\;
}
\ShowLn Generate $\beta_j^*$ by~\eqref{eq:qk**} based on the current $U_{h}^* $, $h=1,\ldots,m$\;
}
\Return  The fiducial samples $\beta_j^*$,    $j=n_{\text{burn}}+1,\ldots,n_{\text{burn}}+n_{\text{mcmc}}$.
\end{algorithm}
\end{minipage}
\end{center}

\begin{remark}
\yf{Algorithm~\ref{alg:gibbs} samples from the marginal generalized fiducial distribution of $\beta$ induced by the joint fiducial construction for $(\beta,\lambda_0)$. The baseline hazard $\lambda_0(t)$ is eliminated through the structure of our DGA, rather than ignored heuristically. Notice that just like Cox's partial likelihood \eqref{eq:CoxPartialL}, the GFD of $\beta$ depends only on the multinomial portion of the likelihood \eqref{eq:CoxCompleteL}. Since this arises naturally from inverting a DGA and not an ad-hoc choice, it can be viewed as a fiducial justification for using the partial likelihood.}
\end{remark}

\begin{remark}
     When the failure times contain ties, the approximate DGA given in Appendix~\ref{s:generalDGA} leads to a fiducial distribution for $\beta$ that is the same as described above with the caveat that if some observations share the failure time $t_k$, they also share the at-risk sets $\mathcal R_k$. This is similar to the approximation of Peto-Breslow method \citep{peto1972contribution,breslow1974covariance}. In fact, the fiducial approach provides a new insight showing that Peto-Breslow approximation is achieved at the cost of Poisson approximation. This approximation works well if the jumps in $\Lambda_0$ are small, i.e., $\max_k m_k\ll m$.
\end{remark}

\begin{remark} \yf{As discussed in Section~\ref{s:FiducialDef},}
   \cite{hannig2016generalized} show that GFD for parametric continuous data usually follows a form of likelihood times a Jacobian. This form is not directly applicable for discrete distributions, because the inverse is a set instead of a point \citep{Hannig2013, liu2016generalized}. When $m\geq p$, the dimension of $\beta$, we derive a likelihood times Jacobian formula for our GFD of $\beta$ in \eqref{eq:JacobianForm} in Appendix~\ref{appB}. 
\end{remark}
\subsection{A conic optimization-based Gibbs sampler}

While the optimization problems in~\eqref{eq:qk*} and \eqref{eq:qk**} can be solved through a brute force search, the computation is usually costly when $\beta$ is multivariate.  
In this subsection, we propose a conic optimization approach to \eqref{eq:qk*} and \eqref{eq:qk**}. \yifan{Conic optimization provides a unifying framework for a wide range of convex optimization problems, including linear, quadratic, semidefinite, and certain nonlinear programs. By formulating the nonlinear constraints as membership in convex (exponential) cones, the conic optimization naturally accommodates our problems.}

\begin{theorem}\label{thm:cone1}
The optimization problem~\eqref{eq:qk*} is equivalent to the following optimization problem, 
\begin{align*}
\min_{\beta,s_1,\ldots,s_m,t_{j,l}} -(\beta^\top X_{i_k} -s_k)\quad &\\
\text{subject to}~~~~
\sum_{j\in \mathcal R_{i_l}} t_{j,l} \leq 1, \quad & l=1,\cdots,m\\
(t_{j,l}, 1, \beta^\top X_j  - s_l) \in K_{exp},\quad  j\in \mathcal R_{i_l},\quad & l=1,\ldots,m\\
\beta^\top X_{i_h} - s_h\geq \log(U^*_h),\quad  h\neq k,\quad & h=1,\ldots,m
\end{align*}
where $K_{exp}$ is an exponential cone defined as
\begin{align*}
K_{exp} := \{ (x_0,x_1,x_2) \in \mathbb R^3: x_0 \geq x_1 \exp(x_2/x_1),\ x_0,x_1\geq 0 \}.
\end{align*}
\end{theorem}

\begin{theorem}\label{thm:cone2}
The optimization problem~\eqref{eq:qk**} is equivalent to the following optimization problem, 
\begin{align*}
\min_{\beta,s_1,\ldots,s_m,t_{j,l}} -\beta^\top w \quad &\\
\text{subject to}~~~ \sum_{j\in \mathcal R_{i_l}} t_{j,l} \leq 1, \quad & l=1,\cdots,m\\
(t_{j,l}, 1, \beta^\top X_j  - s_l) \in K_{exp},\quad  j\in \mathcal R_{i_l},\quad & l=1,\ldots,m\\
\beta^\top X_{i_h} - s_h\geq \log(U_h^*),\quad & h=1,\ldots,m.
\end{align*}
\end{theorem}

The proofs can be found in Appendix~\ref{appA}. The implementation can be done through large-scale optimization software Mosek \citep{mosek2015mosek}.

\section{Theory}\label{sec:theory}

\subsection{\yf{Mode of the generalized fiducial distribution}}

We define for $i=1,\ldots,n,$ 
\begin{equation}\label{eq:qin}
p_i(\beta)=\begin{cases}
\frac{\exp(\beta^\top X_i)}{\sum_{j\in \mathcal R_i}\exp(\beta^\top X_j)} & \mbox{if failure;}\\
1 & \mbox{if censored.}
\end{cases}
\end{equation}
We also define the counting process 
\begin{align*}
N_i(t) =  I\{Y_i \leq t\}\Delta_i,
\end{align*}
and the at-risk process
\begin{align*}
Y_i(t) = I\{Y_i \geq t\}.
\end{align*}

\yifan{First, we show an interesting property of the mode of the fiducial distribution, which in fact, is a maximum likelihood estimator.}
\begin{theorem}\label{lemma}
For a given dataset, any $\beta$ maximizing fiducial probability $P^*(\beta\in Q(U^*))$ is a maximum likelihood estimator.
\end{theorem}
\yifan{
The above result provides a straightforward route to establish consistency for the mode of the fiducial distribution.}
\begin{corollary}
\yifan{Under the conditions that the maximum likelihood estimator is consistent \yf{(e.g., Conditions (2.1)-(2.6) on page 289 of \cite{fleming2013counting};  Conditions (A)-(D) of \cite{andersen1982cox})}, the mode of the fiducial distribution is consistent.}
\end{corollary}

\subsection{A Bernstein–von Mises theorem}

Note that the fiducial distribution is a data-dependent distribution which is defined
for every fixed dataset. The fiducial distribution can be made into a random measure by plugging random variables $(X,Y,\Delta)$ into the observed data. We establish a Bernstein-von Mises theorem
for this random measure for a one-dimensional case. The multivariate case holds similarly as can be seen in Appendix~\ref{appB}.

We first define some notation. 
\yf{Let $\beta^\circ$ and $\lambda_{0}^\circ(t)$ denote the true value of log hazard ratio and baseline hazard function, respectively.
For brevity, this notation is used in this section and proofs in Appendix~\ref{appB}, and is omitted elsewhere when the context remains unambiguous. 
}
Let
$S^j(\beta,t)=\frac{1}{n}\sum_{i=1}^n Y_i(t)X_i^{j}\exp(\beta X_i)$, $j=0,1,2$, and 
\begin{align*}
V(\beta,t) =& \frac{S^2(\beta,t)}{S^0(\beta,t)} - \left(\frac{S^1(\beta,t)}{S^0(\beta,t)}\right)^2\\
=& \frac{1/n\sum_{i=1}^n ( X_i - \epsilon (\beta,t) )^2 Y_i(t)\exp(\beta X_i)}{S^0(\beta,t)},
\end{align*}
where $\epsilon(\beta,t):=\frac{S^1(\beta,t)}{S^0(\beta,t)}$. 
We consider the following regularity conditions:







[1] There exists $\tau$ so that $P(Y_i(\tau)>0)>0$.

[2] There exists a constant $M$, so that $|X_i|<M$.

[3] Let $s^j(\beta,t)$ be the probability limit of $S^j(\beta,t)$.
We assume that 
 \begin{align*}
 H (\beta^\circ)  = \int_0^\tau v(\beta^\circ,t)s^0(\beta^\circ,t) \lambda_0^\circ(t) dt
 \end{align*}
 is positive,
 where $v(\beta,t):=\frac{s^2}{s^0}-e^2$ and $e:=s^1/s^0$.

[4] We assume that 
\begin{align}\label{eq:convegew}
\frac{1}{n}
\sum_{i=1}^{n} \int_0^\tau \frac{| X_i S^0(\beta^\circ,s)- S^1(\beta^\circ,s) |}{S^0(\beta^\circ,s)}  \min_\beta \frac{S^0(\beta,s)}{\exp(\beta X_i)} dN_i(s) \rightarrow 
 \int_0^\tau w(\beta^\circ,s) ds,
\end{align}
in probability, and 
$0< \int_0^\tau w(\beta^\circ,s) ds < \infty$.


Conditions [1]-[3] are typically assumed in the literature for maximum likelihood estimation of the Cox model \citep{andersen1982cox,fleming2013counting}.
\yf{In particular, Condition [1] assumes that there is a non-zero probability that at least some subjects are still alive and uncensored at time $\tau$, which is common in both semiparametric and nonparametric survival analysis \citep{fleming2011counting}.}
Condition [2] assumes the boundedness of covariates, while Condition [3] assumes that the variance is positive.
By Theorem 8.4.1 of \cite{fleming2013counting}, 
Conditions [1]-[3] ensure the regularity conditions given in \cite{andersen1982cox,fleming2013counting}.

Condition [4] assumes that the left hand side of \eqref{eq:convegew} has a limit.
While Bayesian analysis typically requires that the prior distribution is continuous at the true parameter value $\beta^\circ$ and assigns a positive probability at $\beta^\circ$ \citep{kim2003bayesian}, our Condition [4] plays an analogous role in the fiducial framework. 
\yf{In particular, notice that the compensator of
\begin{align*}
\int_0^\tau \frac{| X_i S^0(\beta^\circ,s)- S^1(\beta^\circ,s) |}{S^0(\beta^\circ,s)}dN_i(s) 
\end{align*}
is 
\begin{align*}
 \int_0^\tau | X_i - \epsilon (\beta^\circ,s) | Y_i(s)\exp(\beta^\circ X_i)\lambda_0^\circ(s)ds.
 \end{align*}
Compared to Condition [3] that assumes the variation of covariates in quadratic form, i.e.,  the convergence of
\begin{align*}
\frac{1}{n}
\sum_{i=1}^{n}  \int_0^\tau (X_i - \epsilon (\beta^\circ,s) )^2 Y_i(s)\exp(\beta^\circ X_i)\lambda_0^\circ(s)ds,
\end{align*}
Condition [4] essentially  assumes the variation of covariates in terms of absolute deviation with a weight $\min_\beta \frac{S^0(\beta,s)}{\exp(\beta X_i)}$, i.e., the convergence of 
\begin{align*}
\frac{1}{n}
\sum_{i=1}^{n}  \int_0^\tau |X_i - \epsilon (\beta^\circ,s)|Y_i(s)\exp(\beta^\circ X_i)\lambda_0^\circ(s)\min_\beta \frac{S^0(\beta,s)}{\exp(\beta X_i)}ds.
\end{align*}
The weight satisfies $\frac{1}{n} \leq \min_\beta \frac{1}{n}\sum_{j=1}^n Y_j(s)\exp(\beta X_j)/\exp(\beta X_i)\leq \frac{1}{n}\sum_{j=1}^n Y_j(s)$. Therefore, for $X_i$ away from the boundary of the design space, this weight does not go to zero. 
Thus, Condition [4] ensures that the limit of the discrete analogue of Jacobian \eqref{eq:ContJacobian} is well-behaved in a neighborhood of $\beta^\circ$, thereby allowing the fiducial distribution to concentrate appropriately and enabling asymptotic validity of the proposed inference procedure.}

Let $\eta=\sqrt{n}(\beta-\tilde\beta)$.
We define the rescaled distribution
$r_n(\eta) = r(\tilde \beta+\eta/\sqrt{n})/\sqrt{n}$,
where $r(\beta)$ is the density of the fiducial distribution of $\beta$ generated by Algorithm~\ref{alg:gibbs}. We have the following theorem. 

\begin{theorem}[Asymptotic normality] \label{thm:main}
Under \yifan{Conditions [1]-[4]}, we have 
\begin{align*}
\int |r_n(\eta) - f_N (\eta)| d\eta \rightarrow 0,
\end{align*}
in probability, where $f_N (\eta)$ is the density of a normal distribution with mean $0$ and variance $H^{-1}(\beta^\circ)$.
\end{theorem}

\begin{remark}
\yifan{We note that the proposed fiducial distribution shares the same asymptotic distribution as the maximum likelihood estimator \citep{andersen1982cox} and Bayesian bootstrap estimator \citep{kim2003bayesian}, ensuring asymptotic efficiency under regularity conditions.}
\end{remark}

To conclude, we provide the following corollary which shows that the proposed confidence intervals have asymptotically correct coverage.

\begin{corollary}[Coverage property]\label{thm:coverage}
Under the assumptions in Theorem~\ref{thm:main}, any set $C_{n,\alpha} = \{\beta: ||\beta-\tilde \beta|| \leq \epsilon_{n,\alpha} \}$, \yf{with the radius $\epsilon_{n,\alpha}$ chosen so that $P^*(C_{n,\alpha}) = 1 -\alpha$,} is an $(1 - \alpha)$ asymptotic confidence set for $\beta^\circ$.
\end{corollary}

\begin{remark}
Notice that in the proof of Theorem~\ref{thm:main} we did not use the exact form of Equation~\eqref{eq:qk**} selecting a particular solution out of the random feasible set \eqref{eq:Qfeasable}. Consequently, Theorem~\ref{thm:main} holds for any selection rule. The reason we propose using \eqref{eq:qk**} is based on its performance in small samples. 
\end{remark}

\section{Simulation studies}\label{sec:simulation}

In this simulation, we compare our estimator with the maximum likelihood estimator. We consider \yf{five scenarios} with different combinations of $\beta=(\beta_1,\beta_2)$. \yifan{The implementation code is available at \hyperref[https://github.com/yifan-cui/Semiparametric-fiducial-Cox-models]{https://github.com/yifan-cui/Semiparametric-fiducial-Cox-models}.}

The survival time follows
$$\lambda_T (t) = \lambda_0(t) \exp[\beta_1 X^{(1)}+\beta_2 X^{(2)} ],$$
where the baseline hazard function $\lambda_0(t) = 1$. 
The censoring time is uniformly distributed on $(0,2)$. 
The covariates $X^{(1)}$ and $X^{(2)}$ follow a binomial distribution with success probability $1/2$.

For each scenario, training datasets $(X,Y,\Delta)$ were generated with a sample size $n = 20$. \yifan{Additional simulations are provided in the Supplementary Material.} For each training dataset, we applied our estimator as well as the maximum likelihood estimator \yf{and a Bayesian estimator with a catalytic prior that extends a general class of prior distributions \citep{li2026bayesian}}. The fiducial estimates were based on 1000 iterations of the Gibbs sampler after 100 burn-in iterations.
The simulations were replicated 200 times for each scenario.

We compare the mean squared error (MSE) of point estimators and the coverage and average length of confidence intervals (CI). The numerical results for each scenario are presented in Table~\ref{table:simu}. 
For point estimators, in general, the proposed fiducial method outperforms the maximum likelihood estimator. In particular, the maximum likelihood estimator does not converge in several runs in \yf{Scenarios~1, 2, 4, and 5}. Importantly, when the maximum likelihood estimator fails, the proposed fiducial estimator provides a valid estimate.  For uncertainty quantification, we see that the fiducial confidence interval is comparable to the maximum likelihood confidence interval and, overall, has shorter length. \yf{The Bayesian method outperforms the fiducial method in Scenarios~1 and 2. The proposed fiducial method
outperforms the Bayesian estimator in Scenarios~4 and 5. For Scenario~3, the fiducial method works better in terms of uncertainty quantification and the Bayesian method works better in terms of the point estimator.}

\begin{table}[]
\small
\centering
\caption{Comparison of \yf{frequentist, Bayesian, and fiducial estimators}}
\label{table:simu}
\begin{tabular}{@{}lcccc@{}}
\hline
Model & Estimator &
MSE ($\times 10^{-2}$) &
Length of CI &
Coverage of CI ($\%$)\\
\hline
{Model 1 ($\beta_1=-0.5,\beta_2=0$)} 
          & MLE $\tilde \beta_1$ & 1431$^*$ & 1863.25 & 95 \\
          & \yf{Bayesian $\check \beta_1$} & 45 & 2.54 & 94.5 \\        
          & Fiducial $\hat \beta_1$  &  83 &  3.09 & 93 \\
          & MLE $\tilde \beta_2$ & 286$^*$ &  325.89  & 95\\
          & \yf{Bayesian $\check \beta_2$} & 32 & 2.38  &96.5\\
          & Fiducial $\hat \beta_2$  &   75  & 3.09 & 94.5  \\
{Model 2 ($\beta_1=0,\beta_2=0.5$)}
          & MLE $\tilde \beta_1$ & 300$^*$ &  254.04 & 94   \\
          & \yf{Bayesian $\check \beta_1$} & 33 &  2.08 & 92.5 \\
          & Fiducial $\hat \beta_1$  &  55  & 2.64  & 92.5  \\
          & MLE $\tilde \beta_2$ & 261$^*$ &  271.17  &93 \\
          & \yf{Bayesian $\check \beta_2$} & 25 & 1.99 &95 \\
          & Fiducial $\hat \beta_2$  &  60  & 2.71  &92   \\
{Model 3 ($\beta_1=0.5,\beta_2=1$)} 
          & MLE $\tilde \beta_1$ &  57 & 2.43  & 92.5 \\
          & \yf{Bayesian $\check \beta_1$} & 33 & 1.85 & 87.5 \\
          & Fiducial $\hat \beta_1$  &  53  & 2.48 &  92 \\
          & MLE $\tilde \beta_2$ & 66 &  2.62 &  92.5 \\
          & \yf{Bayesian $\check \beta_2$} & 33 & 1.80 & 86.5 \\
          & Fiducial $\hat \beta_2$  & 61  & 2.64 & 91.5  \\
{Model 4 ($\beta_1=1,\beta_2=1.5$)} 
          & MLE $\tilde \beta_1$ & 49 & 2.47 & 94 \\
          & \yf{Bayesian $\check \beta_1$} & 50 & 1.72 &  74.5\\
          & Fiducial $\hat \beta_1$  &  42 & 2.41 & 94  \\
          & MLE $\tilde \beta_2$ & 730$^*$ & 766.87 & 95.5\\
          & \yf{Bayesian $\check \beta_2$} & 65 & 1.70 &  61\\
          & Fiducial $\hat \beta_2$  &  57   & 2.60 & 94.5  \\
\yf{Model 5 ($\beta_1=1.5,\beta_2=2$)} 
          & \yf{MLE $\tilde \beta_1$} & 458$^*$ & 378.15 & 94 \\
          & \yf{Bayesian $\check \beta_1$} & 99 & 1.64 & 33 \\
          & \yf{Fiducial $\hat \beta_1$}  &  41 & 2.34 & 93.5  \\
          & \yf{MLE $\tilde \beta_2$} & 1669$^*$ & 1570.08& 94\\
          & \yf{Bayesian $\check \beta_2$} & 136 & 1.63 & 9.5\\
          & \yf{Fiducial $\hat \beta_2$}  &  46   & 2.52 & 94.5  \\
\hline
\multicolumn{4}{l}{$*$ indicates MLE did not converge for some runs}
\end{tabular}
\end{table}


\section{Extension to other semiparametric models}\label{sec:extension}

\subsection{Other link functions}

The proposed method and Algorithm~\ref{alg:gibbs} naturally extend to the following model:
\[\lambda(t) = \lambda_0(t) g(\beta^\top X),\]
where $g$ is any well-behaved positive-valued link function.
The large sample theory developed in Section~\ref{sec:theory} also applies to this generalized model.

\subsection{Constrained Cox model}
In certain practical problems, some prior information that restricts model parameters would result in a more interpretable conclusion. Such restrictions cannot be ignored, otherwise the statistical inference may be biased. 
Our fiducial method automatically solves constrained Cox models \citep{ding2015new,yin2021constrained}, for example, with equality constraints  $$\beta \in \{\beta:~ f(\beta) = 0\},$$ and inequality constraints $$\beta \in \{\beta:~ g(\beta) \leq 0\}.$$

\subsection{Additive hazards model}\label{subsec:add}
Aalen’s additive hazards model \citep{aalen1980model,mckeague1986,huffer1991,lin1994semiparametric}
\begin{align*}
\lambda(t) = \beta(t)^\top X,
\end{align*}
provides an alternative to the Cox proportional hazards model.
If we rewrite the model as 
\begin{align*}
\lambda(t) = \beta_0(t)\beta(t)^\top X,
\end{align*}
where $1^\top\beta(t)=1$, the extended fiducial algorithm in Appendix~\ref{s:generalDGA} covers such scenarios if we parametrize $\beta(t)$. 
For example, one might impose a polynomial basis
\begin{align*}
\beta(t)= \alpha_0 + \alpha_1 t + \alpha_2  t^2,
\end{align*}
or a spline model, e.g., a B-spline.
In this case, the probability $\vec{q}_{k}$ in \eqref{eq:Qkl} becomes 
\begin{equation*}
    q_{k,l}(\vec{\alpha})=\begin{cases}
       \frac{\beta^\top(t_k) X_l}{\sum_{j\in \mathcal R_k(t_k)} \beta^\top(t_k) X_j}
       & \mbox{ if $l\in \mathcal R_k(t_k)$};\\
       0 & \mbox{ otherwise}.
    \end{cases}
\end{equation*}

\subsection{Time-varying covariates and coefficients}

Our proposed method also readily extends to the Cox  model with time-varying covariates \citep{fisher1999time},
$$\lambda(t) = \lambda_0(t) \exp(\beta^\top X(t)).$$
For the Cox  model with time-varying coefficients \citep{tian2005cox}, 
$$\lambda(t) = \lambda_0(t) \exp(\beta^\top(t) X(t)),$$
we can parametrize $\beta(t)$ in the same way as Section~\ref{subsec:add} and apply the proposed method to sample from the fiducial distribution and then conduct statistical inference.

\section{Real data application}\label{sec:real}

The HVTN 704/HPTN 085 and HVTN 703/HPTN 081 Antibody Mediated Prevention Phase 2b Prevention Efficacy Trials evaluated the prevention efficacy of an infused monoclonal antibody, VRC01, against the endpoint of HIV diagnosis. Participants were recruited from four continents, primarily Africa and North and South America, for random assignment 1:1:1 to treatment by a low or high dose of VRC01, or placebo. HIV diagnosis rates varied across these populations during the study in all three treatment arms, as did circulating HIV-1 strains and participant characteristics. Evaluating efficacy of a preventative intervention by Cox proportional hazards modeling is a component of the pre-planned statistical analyses of this and similar trials, however, the sample sizes are limited for such models, especially when evaluating efficacy in specific sub-populations. The original study overall found that there was no significant efficacy against diagnosis of HIV-1 disease overall, however when evaluated against diagnosis of strains of HIV-1 that are susceptible to neutralization by VRC01, a pre-specified analysis, the estimated intervention efficacy, pooled across the trials, was 75.4\% (95\% CI 45.5\% to 88.9\%), supporting further research into passive immunoprophylaxis for HIV and the development of HIV-1 vaccines that elicit neutralizing antibodies \citep{doi:10.1056/NEJMoa2031738}.

The preventative efficacy of the VRC01 infusion intervention was not reported for specific sub-populations, for example by analysis within different countries that participated in the study. Here we employed our fiducial Cox analysis for evaluating efficacy within subpopulations that are too small for reliable Cox analysis by standard methodology. In this paper, we conducted a sub-population analysis to evaluate the pooled efficacy of the VRC01 infusion intervention against diagnosis of infection by susceptible HIV-1 in the subset of participants who were recruited at sites in Malawi ($n = 180$, of whom only three were diagnosed with VRC01-susceptible HIV-1) and found that by standard Cox regression analysis (employing the maximum likelihood estimator), the estimator did not converge. Employing the fiducial estimator that we have described here, we found some evidence of a treatment effect: fiducial point estimator for efficacy 82.0\%, one-sided 90\% CI (39.9\%, 100\%), 95\% CI (-4.2\%, 100\%), with fiducial p = 0.053 for efficacy departing from 0\%.
\yifan{The dataset and implementation code are available at \hyperref[https://github.com/yifan-cui/Semiparametric-fiducial-Cox-models]{https://github.com/yifan-cui/Semiparametric-fiducial-Cox-models}.}

\section{Discussion}\label{sec:discussion}

In this paper, we have considered fiducial inference in semiparametric models. Taking the Cox proportional hazards model as a running example, we proposed a novel Gibbs sampler to sample from the fiducial distribution. We have also established the consistency and asymptotic normality of the proposed estimator. In addition, we have also discussed several extensions of our approach to other semiparametric models. Our approach was illustrated via simulation studies and a real data application. \yifan{Our paper contributes to the literature on both fiducial inference and semiparametric inference.}

The proposed method may be extended in several directions.
One possible extension is to consider variable selection in semiparametric models such as the Cox model \citep{fan2002variable,he2020functional} following \cite{williams2019}. It is also possible to consider functional predictors in semiparametric models \citep{chen2011stringing,hao2021semiparametric}.
Another important direction is to consider fiducial approaches to other semiparametric transformation models \citep{cheng1995analysis,zeng2006efficient}
and semiparametric models in causal inference \citep{robins1994estimation,bickel2001inference,laan2003unified} which also has a coarsened data structure.

\section*{Acknowledgement}

The authors are thankful to Quoc Tran Dinh for helpful discussions on optimization.
This research was partially supported by the National Key R\&D Program of China (2024YFA1015600), the National Natural Science Foundation of China (12471266 and U23A2064), the National Institute of Allergy and Infectious Diseases of the National Institutes of Health (UM1AI068635 and R37AI054165), and the National Science Foundation (DMS-1916115, 2113404, 2210337). The content is solely the responsibility of the authors and does not necessarily represent the official views of the National Institutes of Health.



\section*{Appendix}
\begin{appendix}

\section{Proof of Proposition~\ref{prop:compatible}}

\begin{proof}
Notice that based on the DGA, the likelihood that the first failure time is equal to $Y_{i_1}$ is
\[
-\frac{d}{dt} \bar S_1(Y_{i_1}) = \bar \lambda_1(Y_{i_1}) \exp(-\bar\Lambda_1(Y_{i_1})),
\]
and the probability that the subject $i_1$ is the first observed failure is
\[
 \frac{\exp(\beta^\top X_{i_1})}{\sum_{j\in \mathcal R_1(Y_{i_1})} \exp(\beta^\top X_j)}.
\]

Similarly, given the first $k-1$ observed failure times the conditional likelihood of $t_k=Y_{i_k}$ given $t_1=Y_{i_1}, \ldots, t_{k-1}=Y_{i_{k-1}}$ is 
\[
-\frac{d}{dt} \bar S_k(Y_{i_k}) = \bar \lambda_k(Y_{i_k}) \exp(-\bar\Lambda_k(Y_{i_k})),
\]
and the conditional probability that the subject $i_k$ is the $k$-th observed failure is
\[
 \frac{\exp(\beta^\top X_{i_k})}{\sum_{j\in \mathcal R_k(Y_{i_k})} \exp(\beta^\top X_j)}.
\]

The joint likelihood implied by our DGA is 
\begin{equation}\label{eq:preLikelihood}
 \prod_{k=1}^m  \frac{\exp(\beta^\top X_{i_k})}{\sum_{j\in \mathcal R_k(Y_{i_k})} \exp(\beta^\top X_j)}\left[\lambda_0(Y_{i_k}){\sum_{j\in \mathcal R_k(Y_{i_k})} \exp(\beta^\top X_j)} \right]\exp\left(-\bar \Lambda_k(Y_{i_k})\right).
\end{equation}
By combining integrals together and redistributing the sums over the at-risk sets we get
\[
 \exp\left(-\sum_{k=1}^m \bar\Lambda_k(Y_{i_k})\right) = \prod_{i=1}^n S_i(Y_i).
\]
By rearranging terms in \eqref{eq:preLikelihood} we get \eqref{eq:CoxCompleteL} as desired.
\end{proof}

\section{Fiducial inversion for baseline hazard}\label{s:BaselineHazard}
In this section, we will continue the inversion process described in Section~\ref{s:DGAinversion} to derive GFD for the baseline hazard in the Cox model. 

Recall that $t_1,\ldots,t_m$ denote the ordered failure times. Set $t_0=0$ and $t_{m+1}=\infty$ to simplify notation.  Given $\beta^*$ a sample from GFD for $\beta$, the inverse mapping of the part of the DGA concerning $\lambda_0$ is
\begin{equation}\label{eq:Qlambda}
  Q_{\beta^*}=\{\lambda_0(t) : t_k=\bar S_k^{-1}(W_k^*),\  k=1,\ldots, m+1\},
\end{equation}
where $W_k^*$ are i.i.d.~$\text{Uniform}(0,1)$, and $\beta$ is replaced by the fiducial sample $\beta^*$ in the definition of $\bar S_k$ as introduced in Section~\ref{s:basicDGA}.

For simplicity of calculation, the version of $\lambda_0$ we select is piecewise constant 
\[
\lambda_{0}^*(t)= \sum_{k=1}^{m+1} \lambda^*_k I_{(t_{k-1},t_k]}(t).
\]
Notice that
for any $k=1,\ldots,m$, the equation $t_k=\bar S_k^{-1}(W_k^*)$ is satisfied if and only if 
\[
\lambda_k^*=-\log(W_k^*)/l_k\sim \mbox{Exp}(l_k),
\] where  $l_k=\sum_{i=1}^n \left(t_k\wedge Y_i-t_{k-1}\wedge Y_i\right)\exp(\beta^{*\top} X_i)$.
Finally, consider $k=m+1$. Since we did not observe the $m+1$-th failure, we only get partial information. In particular, $t_{m+1}=\bar S_{m+1}^{-1}(W_{m+1}^*)$ implies $\lambda_{m+1}^*\leq-\log(W_{m+1}^*)/l_{m+1}$; we recommend using $\lambda_{m+1}^*\sim\mbox{Exp}(l_m\vee 2l_{m+1})$.

The fiducial sampler for the baseline hazard is summarized in Algorithm~\ref{alg:baseline}.

\begin{center} 
\begin{minipage}{\linewidth}
\begin{algorithm}[H] 
\SetAlgoLined
\caption{A fiducial sampler for the baseline hazard \label{alg:baseline}}
\KwIn{Dataset  $(X_i, Y_i, \Delta_i)$,
fiducial samples $\beta_j^*$   
}
\ShowLn Calculate $l_k=\sum_{i=1}^n \left(t_k\wedge Y_i-t_{k-1}\wedge Y_i\right)\exp(\beta^{*\top} X_i), k=1,\ldots,m+1$;\\
\ShowLn Generate independent $\lambda_k^* \sim \mbox{Exp}(l_k),k=1,\ldots,m,$ and $\lambda_{m+1}^* \sim \mbox{Exp}(l_m\vee 2l_{m+1})$;\\
\ShowLn Generate $\lambda_{0,j}^*(t)= \sum_{k=1}^{m+1} \lambda^*_k I_{(t_{k-1},t_k]}(t);$\\
\Return  The fiducial samples $\lambda^*_{0,j}(t)$.
\end{algorithm}
\end{minipage}
\end{center}

Because the fiducial inversion does not use any smoothness assumptions on $\lambda_0(t)$, the fiducial distribution of $\lambda_0^*(t)$ is very rough and does not concentrate near the true distribution. However, using similar arguments as in \cite{cuihannig2019}, the fiducial distribution of the cumulative baseline hazard $\Lambda_0^*(t)=\int_0^t \lambda_0^*(s)\,ds$ satisfies a Bernstein-von Mises theorem, and concentrates near the true cumulative baseline hazard. 
Finally, we remark that the confidence intervals for the cumulative baseline hazard were very conservative. This is not surprising as the proposed fiducial inference concentrates its power on estimating $\beta$. 

\begin{remark}
\yifan{As we have seen in inverting the DGA, the fact that $\lambda_0$ can be chosen freely implies that the failure times do not carry any substantial information about $\beta$.}
If additional restrictions were placed on the baseline hazard, e.g., some level of smoothness, then the inversion \eqref{eq:Qlambda} may not exist for all $W_j^*$, $j=1,\ldots,m+1$. This in turn would change the GFD of $\beta^*$ by reweighting it by the probability that the inversion exists making it no longer related to the partial likelihood of the Cox model. This insight provides a new understanding of the fact that 
the Cox model is based on a likelihood function that eliminates the nuisance parameter $\lambda_0(t)$ which can be arbitrary, leaving a function that depends only on the regression coefficients of interest \citep{cox1972regression}. 
\end{remark}

\section{Data generating algorithm}\label{s:generalDGA}

In this section, we discuss a generalization of the DGA in Section~\ref{s:basicDGA} to more general survival models.
Recall that for each subject $i=1, \ldots, n$, the cumulative hazard is $\Lambda_i$, the subject's survival function is $S_i(t)=\exp(-\Lambda_i(t))$, and their potentially counterfactual censoring time is $c_i$. The subjects are assumed to be independent.  

We proceed by iteratively generating the failure times $t_k$ and the set of subjects that failed at that time $d_k$, $k=1,\ldots,K$. 
We will denote by $m_k=|d_k|$ the number of failures at time $t_k$; notice that $m=\sum_{k=1}^K m_k$.
The $k$-th failure removed set $\bar{\mathcal{R}}_k=\{1,\ldots,n\}\setminus \bigcup_{l=1}^{k-1} d_l$, 
and 
\[\bar S_k(t)=\prod_{i\in \bar{\mathcal R}_k} \frac{S_i((t_{k-1}\vee t)\wedge c_i)}{S_i(t_{k-1}\wedge c_i)},\] where $t_0=0$.
Notice that the freezing of the survival function at the censoring times used together with the failure removed set has the same effect as using the usual at-risk set. 

The $k$-th failure time is generated by 
$t_k=\bar S_k^{-1}(W_k)$, where $W_k$ are i.i.d.~$\text{Uniform}(0,1)$. 
Next, we need to generate which subjects $d_k$ failed at time $t_k$.
To this end, let $B_{k,i}^{dt},\ i\in\bar{\mathcal R}_k$ be independent Bernoulli$(\vec{q_k}^{dt})$, conditioned on the event $\{\sum_{i\in\bar{\mathcal R}_k} B_{k,i}^{dt}\geq 1\}$, where
$q_{k,i}^{dt}=\frac{S_i((t_k-dt)\wedge c_i)- S_i(t_k\wedge c_i)}{S_i((t_k-dt)\wedge c_i)}$.
Denote by $B_{k,i}$ the limiting distribution of $B_{k,i}^{dt}$ as $dt\to 0$. The set of subjects that failed at time $t_k$ is generated by sampling $B_{k,i}$ and setting $d_k=\{i: B_{k,i}=1\}$. 
This process is continued until either  $\bar{\mathcal R}_{k+1}=\emptyset$, or the generated failure time $t_k=\infty$.

When $\bar S_k(t)$ is continuous at $t_k$, the limiting distribution 
has only one failure with probability one which is selected from the 
multinomial$(1,\vec{q}_{k})$, where
\begin{equation}\label{eq:generalQ}
q_{k,i} = \begin{cases} \frac{\frac d{dt}\Lambda_i(t_k)}{\sum_{j\in\mathcal R_k(t_k)} \frac d{dt}\Lambda_j(t_k)} & \mbox{ if $i \in\mathcal R_k(t_k)$};\\
   0 & \mbox{otherwise},
   \end{cases}
\end{equation}
where $\mathcal R_k(t_k)$ is the at risk set defined in \eqref{eq:atRiskSet}.
Thus if $\Lambda_i(t) = \Lambda_0(t) \exp(\beta^\top X_i)$, the multinomial probability $\vec{q}_{k}$ is the same as in \eqref{eq:Qkl}. Additionally, if $\Lambda_0(t)=\int_0^t \lambda_0(s)\,ds$ we have exactly the same DGA as in Section~\ref{s:basicDGA}.

When $\bar S_k(t)$ has a jump at $t_k$, then we can have more than one failure at that time. For any $d_k\subset\mathcal R_k (t_k)$, the probability of generating these failures is given by
\begin{equation}\label{eq:generalJumpP}
P(d_k \mbox{ is selected})=
 \frac{\bar S_k(t_k^-)}{\bar S_k(t_k^-)-\bar S_k(t_k)}\,
 \prod_{i\in d_k}\frac{S_i(t_{k}^-)-S_i(t_{k})}{S_i(t_{k}^-)}
 \prod_{i\in {\mathcal R}_{k}(t_k)\setminus d_k}
 \frac{S_i(t_{k})}{S_i(t_{k}^-)},
\end{equation}
where $S_i(t^-)$ denotes the left limit of the survival function. 

Using an argument similar to the proof of Proposition~\ref{prop:compatible}, we can show that this DGA produces the same likelihood as generating each failure time individually. 
However, even if $\Lambda_i(t) = \Lambda_0(t) \exp(\beta^\top X_i)$, the probability \eqref{eq:generalJumpP} does not provide the usual simple partial likelihood. Therefore, following the usual practice of using approximate likelihood \citep{peto1972contribution,kalbfleisch1973marginal,breslow1974covariance,efron1977efficiency}, we propose an approximate DGA for this setting:

First, we approximate the probability 
\[
\frac{S_i(t_{k}^-)-S_i(t_{k})}{S_i(t_{k}^-)}
\approx (\Lambda_0(t_k)- \Lambda_0(t_k^-))\exp(\beta^\top X_i).
\]
Next, at each time $t_k$, we approximate the distribution of the number of failures 
$m_k=\sum_{i\in\mathcal R_k(t_k)} B_{k,i}$ using the Poisson$(\eta_k)$ distribution conditional on the set $\{m_k\geq 1\}$, where 
\[
 \eta_k=\sum_{i\in\mathcal R_k(t_k)} (\Lambda_0(t_k)- \Lambda_0(t_k^-))\exp(\beta^\top X_i).
\]
Finally, given $m_k$, the set $d_k$ is generated from the multinomial$(m_k,\vec{q}_k)$ distribution
conditional on the event that each category is observed at most once, with  $\vec{q}_{k}$ given by \eqref{eq:Qkl}.

\section{Conic optimization}\label{appA}

In this section, we provide proofs of Theorems~\ref{thm:cone1} and \ref{thm:cone2}.
For the problem~\eqref{eq:qk*}, we aim to 
\begin{align*}
\max_{\beta} \frac{\exp(\beta^\top X_{i_k})}{\sum_{j \in \mathcal R_{i_k}} \exp(\beta^\top X_j)}
\end{align*}
subject to $U^*_h\leq q_h(\beta)$ for any $h\neq k$.
It is equivalent to 
\begin{align*}
\max_{\beta} \beta^\top X_{i_k} - \log\left(\sum_{j\in \mathcal R_{i_k}} \exp(\beta^\top X_j)\right)
\end{align*}
subject to $\log(U^*_h)\leq \beta^\top X_{i_h} - \log(\sum_{j\in \mathcal R_{i_h}} \exp(\beta^\top X_j))$ for any $h\neq k$.

By introducing decision variables $s_l$, it is further equivalent to
\begin{align*}
\min_{\beta,s_1,\ldots,s_m} -(\beta^\top X_{i_k} -s_k) \quad &\\
\text{subject to}~~~\log\left(\sum_{j\in \mathcal R_{i_l}} \exp(\beta^\top X_j)\right) \leq s_l, \quad & l=1,\cdots,m\\
\beta^\top X_{i_h} - s_h\geq \log(U^*_h),\quad 
h\neq k, \quad & h=1,\cdots,m.
\end{align*}

By introducing decision variables $t_{j,l}$, the optimization becomes 
\begin{align*}
\min_{\beta,s_1,\ldots,s_m,t_{j,l}} -(\beta^\top X_{i_k} -s_k) \quad &\\
\text{subject to}~~~
\sum_{j\in \mathcal R_l} t_{j,l} \leq 1, \quad & l=1,\cdots,m\\
\exp(\beta^\top X_j - s_l) \leq t_{j,l},\quad  j\in \mathcal R_l,\quad & l=1,\ldots,m\\
\beta^\top X_{i_h} - s_h\geq \log(U^*_h),\quad 
h\neq k, \quad & h=1,\cdots,m.
\end{align*}

The final optimization problem becomes 
\begin{align*}
\min_{\beta,s_1,\ldots,s_m,t_{j,l}} -(\beta^\top X_{i_k} -s_k) \quad &\\
\text{subject to}~~~\sum_{j\in \mathcal R_{i_l}} t_{j,l} \leq 1, \quad & l=1,\cdots,m\\
(t_{j,l}, 1, \beta^\top X_j  - s_l) \in K_{exp},\quad   j\in \mathcal R_{i_l},\quad & l=1,\ldots,m\\
\beta^\top X_{i_h} - s_h\geq \log(U^*_h),\quad 
h\neq k, \quad & h=1,\cdots,m.
\end{align*}

For the problem~\eqref{eq:qk**}, 
we aim to 
\begin{align*}
\max_{\beta} \beta^\top w
\end{align*}
subject to $U^*_h\leq q_h(\beta)$ for $h=1,\ldots, m$.
It is equivalent to
\begin{align*}
\max_{\beta} \beta^\top w
\end{align*}
subject to $\log(U^*_h)\leq \beta^\top X_{i_h} - \log(\sum_{j\in \mathcal R_{i_h}} \exp(\beta^\top X_j))$ for any $h$.

By introducing decision variables $s_l$, it is further equivalent to
\begin{align*}
\min_{\beta,s_1,\ldots,s_m} -\beta^\top w \quad & \\
\text{subject to}~~~ \log\left(\sum_{j\in \mathcal R_{i_l}} \exp(\beta^\top X_j)\right) \leq s_l, \quad & l=1,\cdots,m\\
\beta^\top X_{i_h} - s_h\geq \log(U^*_h),\quad &
h=1,\cdots,m.
\end{align*}

By introducing decision variables $t_{j,l}$, the optimization becomes 
\begin{align*}
\min_{\beta,s_1,\ldots,s_m,t_{j,l}} -\beta^\top w \quad &\\
\text{subject to}~~~\sum_{j\in \mathcal R_{i_l}} t_{j,l} \leq 1, \quad & l=1,\cdots,m\\
\exp(\beta^\top X_j - s_l) \leq t_{j,l},\quad  j\in \mathcal R_{i_l},\quad & l=1,\ldots,m\\
\beta^\top X_{i_h} - s_h\geq \log(U^*_h), \quad & h=1,\ldots,m.
\end{align*}

The final optimization problem becomes 
\begin{align*}
\min_{\beta,s_1,\ldots,s_m,t_{j,l}} -\beta^\top w\quad &\\
\text{subject to}~~~\sum_{j\in \mathcal R_{i_l}} t_{j,l} \leq 1, \quad & l=1,\cdots,m\\
(t_{j,l}, 1, \beta^\top X_j  - s_l) \in K_{exp},\quad   j\in \mathcal R_{i_l},\quad & l=1,\ldots,m\\
\beta^\top X_{i_h} - s_h\geq \log(U^*_h),\quad & h=1,\ldots,m.
\end{align*}

\section{Proofs}\label{appB}

\begin{proof}[Proof of Theorem~\ref{lemma}]

Recall that the maximum likelihood estimator maximizes 
\begin{align*}
    \prod_{i=1}^n \left[ \frac{\exp(\beta^\top X_i)}{\sum_{j\in \mathcal R_i}\exp(\beta^\top X_j)} \right]^{\Delta_i} = \prod_{i=1}^n p_i(\beta),
\end{align*}
where $p_i$ are defined in \eqref{eq:qin}.
Also recall that $Q(U^*)=\{\beta: \beta\mbox{ satisfying \eqref{eq:LawrenceInv2}}\}$.
So we have that the fiducial probability
\[
P^*(\beta\in Q(U^*)) \propto \prod_{i=1}^n  p_i(\beta).
\]

By Section~2.3 of \cite{andersen1982cox} and Theorem~8.3.1 of \cite{fleming2013counting}, the mode of the fiducial distribution is consistent as $\tilde \beta$ is consistent.

\end{proof}

\yf{If all subjects in the corresponding at-risk set $\mathcal R_{i}$ have the same covariate vector, then $p_i(\beta)=|\mathcal R_{i}|^{-1}$ is independent of $\beta$. This observation contributes no information about $\beta$ and can therefore be omitted from the calculations below. Thus in what follows, all products and sums involving the Jacobian term are taken only over $i$ for which $p_i(\beta)$ is non-constant.}

\begin{proof}[Proof of Theorem~\ref{thm:main}]

We start with one dimensional $\beta$.
{We omit the argument $\beta$ in the expression below unless we need to specify it.} Notice that the optimal solution of the problem \eqref{eq:qk**} will with probability one have exactly one of its constraints active, i.e., for exactly one $k\in\{1,\ldots,m\}$, 
\[
U_k^* = \frac{\exp(\beta X_{i_k})}{\sum_{j\in \mathcal R_{i_k}} \exp(\beta X_j)}=p_{i_k}(\beta),
\]
and for the others $l\neq k$,
\[
U_l^* < \frac{\exp(\beta X_{i_l})}{\sum_{j\in \mathcal R_{i_l}} \exp(\beta X_j)}=p_{i_l}(\beta).
\]
We will use this observation to derive a fiducial density $r(\beta)$.

Set for $i=1,\ldots,n$,
\[
 \bar{r}_i(\beta) = \frac{|J_i(\beta)|}{p_i} \prod_{j=1}^{n} p_j,
\]
where 
\[
 |J_i(\beta)| = \left|\frac{\partial p_i(\beta)}{\partial \beta}\right|
 = \frac{|\exp(\beta X_i)X_i\sum_{j\in \mathcal R_i} \exp(\beta X_j) - \exp(\beta X_i)\sum_{j\in \mathcal R_i} X_j\exp(\beta X_j)  |}{[\sum_{j\in \mathcal R_i} \exp(\beta X_j)]^2}.
\]
Using the counting process, we write the fiducial distribution as
\begin{equation}\label{eq:JacobianForm}
 r(\beta) \propto
 \sum_{i=1}^n
\int_0^\tau c_i^{-1}\bar{r}_i(\beta)dN_i(s)
 =
 \prod_{j=1}^{n} p_j \left(\sum_{i=1}^{n} \int_0^\tau\frac{|J_i|}{p_i c_i} dN_i(s)\right),
\end{equation}
where $c_i= \max_{\beta}(p_i)-\min_{\beta}(p_i)$ for failures and $c_i=1$ for censored observations.

We expand $\log(L_n(\beta))$ at the maximum likelihood estimator $\tilde \beta$,
\[
\log(L_n(\beta)) = \log L_n(\tilde\beta)+\frac{1}{2}\frac{\partial^2 \log(L_n(\beta))}{\partial \beta^2}|_{\beta=\beta'}\,
(\beta - \tilde \beta)^2,
\]
where $\beta'$ is on the line segment between $\beta$ and $\tilde\beta$.

We define unscaled $$\tilde r(\beta)=\frac{1}{n^2} \prod_{j=1}^{n} p_j \left(\sum_{i=1}^{n} \int_0^\tau \frac{|J_i|}{p_i c_i} dN_i(s)\right),$$ and $$\tilde r_n(\eta)=\tilde r(\tilde \beta+\eta/\sqrt{n}).$$
We prove our theorem by establishing the following 
two results:

(i) First, in the following, we show that
\begin{align*}
\log \tilde r_n(\eta)-\log L_n(\tilde\beta)  \rightarrow -\frac{H(\beta^\circ)}{2}\eta^2 +\log\left(\int_0^\tau w(\beta^\circ,s) ds\right),
\end{align*}
in probability.

If we parametrize $\beta_n=\tilde \beta+\eta/\sqrt{n}$, we have both $\tilde\beta$ and $\beta'$ converges to $\beta^\circ$.
By Theorem~3.2 of \cite{andersen1982cox}, 
we have that
\begin{align*}
-\frac{1}{n}\frac{\partial^2 \log(L_n(\beta))}{\partial \beta^2}|_{\beta=\beta'}
\rightarrow H(\beta^\circ),
\end{align*}
in probability. By a simple calculation, we have that
\begin{align*}
& \sum_{i=1}^{n} \int_0^\tau\frac{|J_i|}{p_i c_i} dN_i(s)
\\ =&  \sum_{i=1}^{n} \int_0^\tau \frac{|\exp(\beta X_i)X_i\sum_{j} Y_j(s)\exp(\beta X_j) - \exp(\beta X_i)\sum_{j} Y_j(s)X_j\exp(\beta X_j)  |}{[\sum_{j} Y_j(s)\exp(\beta X_j)]^2}  \frac{\sum_{j} Y_j(s)\exp(\beta X_j)}{\exp(\beta X_i)} \\ &\times \left[ \max_{\beta}(p_i)-\min_{\beta}(p_i)\right]^{-1}dN_i(s)
\\ =&  \sum_{i=1}^{n} \int_0^\tau \frac{|X_i\sum_{j} Y_j(s) \exp(\beta X_j) - \sum_{j} Y_j(s) X_j\exp(\beta X_j)  |}{[\sum_{j}Y_j(s) \exp(\beta X_j)]}   \left[ \max_{\beta}(p_i)-\min_{\beta}(p_i)\right]^{-1} dN_i(s)
\\ =&   \sum_{i=1}^{n} \int_0^\tau \frac{|\sum_{j\neq i} (X_i-X_j )Y_j(s)\exp(\beta X_j)   |}{[\sum_{j} Y_j(s)\exp(\beta X_j)]}  \left[ \max_{\beta}(p_i)-\min_{\beta}(p_i)\right]^{-1} dN_i(s)
\\ =& \sum_{i=1}^{n} \int_0^\tau \frac{|\sum_{j} (X_i-X_j )Y_j(s)\exp(\beta X_j)   |}{[\sum_{j} Y_j(s)\exp(\beta X_j)]}   \left[ \max_{\beta}(p_i)-\min_{\beta}(p_i)\right]^{-1} dN_i(s).
\end{align*}

Note that $\min_{\beta}(p_i)=0$ for failure observations. 
By \yifan{Condition [4]} and Lemma~\ref{lemma:bounded},
\begin{align*}
\frac{1}{n}
\sum_{i=1}^{n} \int_0^\tau \frac{| X_i S^0(\beta_n,s)- S^1(\beta_n,s) |}{S^0(\beta_n,s)}  \min_\beta \frac{S^0(\beta,s)}{\exp(\beta X_i)} dN_i(s) \rightarrow 
 \int_0^\tau w(\beta^\circ,s) ds,
\end{align*}
in probability. 
Therefore, we have
\begin{align*}
 \frac{1}{n^2}\sum_{i=1}^{n} \int_0^\tau\frac{|J_i(\beta_n)|}{p_i(\beta_n) c_i} dN_i(s) \rightarrow 
 \int_0^\tau w(\beta^\circ,s) ds,
\end{align*}
in probability. 
So we have that 
\begin{align*}
\log \tilde r_n(\eta)-\log L_n(\tilde\beta)  \rightarrow -\frac{H(\beta^\circ)}{2}\eta^2 + \log\left(\int_0^\tau w(\beta^\circ,s) ds\right),
\end{align*}
in probability, which completes the proof of (i).

(ii) Next, we will show that 
\begin{align*}
\int \tilde r_n(\eta)/L_n(\tilde \beta)d\eta \rightarrow \sqrt{\frac{2\pi}{H(\beta^\circ)}} \int_0^\tau w(\beta^\circ,s) ds,
\end{align*}
in probability.
We define $A_n=\{\eta:|\eta+\sqrt{n}(\tilde\beta-\beta^\circ)|\leq \sqrt n\delta_0\}$.


Note that           
\begin{align*}
\int \tilde r_n(\eta)/L_n(\tilde \beta)d\eta=
\int_{A_n} \tilde r_n(\eta)/L_n(\tilde \beta)d\eta + \int_{A_n^c} \tilde r_n(\eta)/L_n(\tilde \beta)d\eta.
\end{align*}
By uniform law of large numbers, for $\epsilon =H(\beta^\circ)/2$, there exists $\delta'>0$ so that 
$|\beta-\beta^\circ|<\delta'$ implies 
$|H(\beta)-H(\beta^\circ)|<\epsilon$, i.e., 
\begin{align*}
\sup_{|\beta-\beta^\circ|<\delta'} |H(\beta)-H(\beta^\circ)|<\epsilon.
\end{align*}
Taking $\delta=\min(\delta',\delta_0)$, we have that
\begin{align*}
\sup_{|\beta-\beta^\circ|\leq \delta} \frac{1}{n}
\frac{\partial^2 \log(L_n(\beta))}{\partial \beta^2}|_{\beta=\beta} < -\frac{H(\beta^\circ)}{2}.
\end{align*}
By the dominated convergence theorem, we have that
$$\int_{A_n} \tilde r_n(\eta)/L_n(\tilde \beta)d\eta \rightarrow \sqrt{\frac{2\pi}{H(\beta^\circ)}} \int_0^\tau w(\beta^\circ,s) ds,$$
in probability.

For $\int_{A_n^c} \tilde r_n(\eta)/L_n(\tilde \beta)d\eta$, recall that 
$L_n(\beta) = \prod_{i=1}^{n}p_i$, without loss of generality, suppose that $\min_{i \in \mathcal R_1} X_{i}<X_1<\max_{i \in \mathcal R_1} X_{i}$. Then we have 
\begin{align*}
p_1= \frac{\exp(\beta X_{1})}{\sum_{j\in \mathcal R_1} \exp(\beta X_j)}= \frac{1}{\sum_{j\in \mathcal R_1} \exp(\beta (X_j-X_1))},
\end{align*}
is integrable. By concavity of the log-likelihood \citep{kim2003bayesian}, for any $\delta>0$, there exists an $\epsilon>0$ such that
\begin{align*}
P\left(\sup_{|\beta-\beta^\circ|>\delta}
\left(\log \frac{L_n(\beta)}{p_1(\beta)}- \log \frac{L_n(\beta^\circ)}{ p_1(\beta^\circ)}\right)
\leq -n\epsilon \right)\rightarrow 1.
\end{align*}
Therefore, by Lemma~\ref{lemma:bounded},
\begin{align*}
\int_{A_n^c} \tilde r_n(\eta)/L_n(\tilde \beta)d\eta
=O\left( \int_{A^c} p_1 \prod_{i=2}^n  p_i /p_i(\beta^\circ) d\beta\right)
\rightarrow 0,
\end{align*}
in probability. 
Therefore, we have that
\begin{align*}
\int \tilde r_n(\eta)/L_n(\tilde \beta)d\eta \rightarrow \sqrt{\frac{2\pi}{H(\beta^\circ)}} \int_0^\tau w(\beta^\circ,s) ds.
\end{align*}

Combining (i) and (ii), by \yifan{Theorem~21} of  \cite{ferguson1996course}, we have that 
\begin{align*}
\int |r_n(\eta) - f_N (\eta)| d\eta \rightarrow 0,
\end{align*}
in probability, where $f_N (\eta)$ is the density of normal with mean $0$ and variance $H^{-1}(\beta^\circ)$.
 \end{proof}

 \begin{remark}
For a $d$-dimensional $\beta$, $v(\beta,t)$ in \yifan{Condition [3]} is replaced by $v(\beta,t)=\frac{\mathbf{s}^2}{s^0}-e^{\otimes 2}$, where $\mathbf{s}^2$ is the limit of
\begin{align*}
\frac{1}{n}\sum_{i=1}^n Y_i(t)X_i^{\otimes 2}\exp(\beta X_i),
\end{align*}
and $\otimes$ is an outer product. Moreover, for any $\mathbf{i}=(i_1,\cdots,i_d)$, we replace Condition~[4] by

$[4']$ $\frac{1}{n}\sum_{i=1}^n\int_0^\tau \frac{|det(J_{\mathbf{i}}(\beta))|}{\prod_{j\in \mathbf{i}} p_j c_{\mathbf{i}}} dN_i(s) \rightarrow \int_0^\tau w(\beta,s) ds$ in probability.

\noindent Under these conditions the analogue of Theorem~\ref{thm:main} holds. 
\end{remark}

\begin{proof}
Note that
\begin{align*}
r(\beta)  \propto \sum_{\mathbf{i}=(i_1,\cdots,i_d)} 
\prod_{l=1}^n p_l
\frac{|det(J_{\mathbf{i}}(\beta))|}{\prod_{j\in \mathbf{i}} p_j c_{\mathbf{i}}},
\end{align*}
where 
$J_{\mathbf{i}}(\beta)= \nabla_\beta \,p_{\mathbf i}$
is a $d\times d$ matrix. A similar result of Theorem~\ref{thm:main} holds under Conditions [1]-[3], and [4'].
\end{proof}

\begin{lemma}\label{lemma:bounded}
We have
$0<\frac{1}{n^2}\sum_{i=1}^{n} \int_0^\tau\frac{|J_i|}{p_i c_i} dN_i(s)\leq M$ almost surely for some $M>0$.
\end{lemma}

\begin{proof}[Proof of Lemma~\ref{lemma:bounded}]
Recall that
\begin{align*}
& \frac{1}{n^2}\sum_{i=1}^{n} \int_0^\tau\frac{|J_i|}{p_i c_i} dN_i(s)
\\ =& \frac{1}{n^2}\sum_{i=1}^{n} \int_0^\tau \frac{|\sum_{j} (X_i-X_j ) Y_j(s) \exp(\beta X_j)   |}{[\sum_{j} Y_j(s) \exp(\beta X_j)]}   \left[ \max_{\beta}(p_i)-\min_{\beta}(p_i)\right]^{-1} dN_i(s)\\
=& \frac{1}{n}\sum_{i=1}^{n} \int_0^\tau \frac{|\sum_{j} (X_i-X_j )Y_j(s)\exp(\beta X_j)   |}{[\sum_{j} Y_j(s)\exp(\beta X_j)]}   \min_\beta \frac{S^0(\beta,s)}{\exp(\beta X_i)} dN_i(s).
\end{align*}

Note that we have  
\begin{align*}
\left|\sum_{j\neq i} (X_i-X_j )Y_j(s)\exp(\beta X_j)\right|  \neq 0
\end{align*}
 almost surely, 
 \begin{align*}
\frac{|\sum_{j} (X_i-X_j )Y_j(s)\exp(\beta X_j)   |}{[\sum_{j} Y_j(s)\exp(\beta X_j)]} \leq 
\max_{i,j} |X_i-X_j|  \frac{[\sum_{j} Y_j(s)\exp(\beta X_j)   ]}{[\sum_{j} Y_j(s)\exp(\beta X_j)]} 
\end{align*}
is bounded, 
and 
 \begin{align*}
\min_\beta \frac{\frac{1}{n}\sum_{i=1}^n Y_j(s)\exp(\beta X_j)}{\exp(\beta X_i)} \leq 1,
\end{align*}
which completes the proof.
\end{proof}

\begin{proof}[Proof of Corollary~\ref{thm:coverage}]
We know that $n^{1/2} (\tilde \beta-\beta^\circ) \rightarrow N(0, H^{-1}(\beta^\circ))$ in distribution and Theorem~\ref{thm:main} implies that $n^{1/2} (\beta^*-\tilde \beta) \rightarrow N(0, H^{-1}(\beta^\circ))$ in distribution in probability.
So we have
\begin{align*}
1-\alpha = P^*(\{\beta: ||\beta-\tilde \beta|| \leq \epsilon_{n,\alpha} \})=
P^*(\{\beta: n^{1/2}||\beta-\tilde \beta|| \leq n^{1/2}\epsilon_{n,\alpha} \})
\end{align*}
converges to $\Gamma(\epsilon_\infty)$, where $\Gamma$ is the cumulative distribution function of the limit of  $n^{1/2}||\beta^*-\tilde \beta||$ and $\epsilon_\infty$ is the unique limit of $n^{1/2}\epsilon_{n,\alpha}$.
Therefore, we have that
\begin{align*}
P(\beta^\circ \in \{\beta: ||\beta-\tilde \beta|| \leq \epsilon_{n,\alpha} \}) = 
P(||\beta^\circ-\tilde \beta|| \leq \epsilon_{n,\alpha} ) \rightarrow \Gamma(\epsilon_\infty)=1-\alpha.
\end{align*}
\end{proof}

\end{appendix}

\linespread{1.15}\selectfont

\bibliographystyle{asa}
\bibliography{reference}

\end{document}